\documentclass[a4paper,UKenglish]{article}

\usepackage{microtype}

\usepackage{amsmath}
\usepackage{amssymb}
\usepackage{proof}
\usepackage[all]{xy}
\usepackage{graphicx}

\newcommand{\calC}{\mathcal{C}}
\newcommand{\calE}{\mathcal{E}}

\newcommand{\calP}{\mathcal{P}}
\newcommand{\calQ}{\mathcal{Q}}

\usepackage{amsthm}
\theoremstyle{definition}
\newtheorem{definition}{Definition}
\newtheorem{remark}{Remark}
\newtheorem{example}{Example}
\theoremstyle{theorem}
\newtheorem{lemma}{Lemma}
\newtheorem{theorem}{Theorem}

\newcommand{\dom}{\mathrm{dom}}	
\newcommand{\nats}{\mathbb{N}}

\newcommand{\sop}{[}
\newcommand{\scl}{]}
\newcommand{\sel}[2]{#1 \backslash #2}
\newcommand{\unsubst}[2]{\sop \sel{#1}{#2} \scl}
\newcommand{\Var}{\mathrm{V}}

\usepackage{mathabx}
\newcommand{\impl}{\rightarrow}
\newcommand{\dual}[1]{\overline{#1}}
\newcommand{\Lor}{\bigvee}

\newcommand{\EV}{\mathrm{EV}}	

\newcommand{\Shallow}{\mathrm{Sh}}
\newcommand{\Deep}{\mathrm{Dp}}
\newcommand{\Seq}{\mathrm{Seq}}
\newcommand{\Exp}{\mathrm{Exp}}
\newcommand{\weight}{\mathrm{w}}
\newcommand{\merge}{\cup}

\newcommand{\omerge}{\sqcup}
\newcommand{\oMerge}{\bigsqcup}
\newcommand{\pred}{\mathrm{pred}}

\newcommand{\mergered}{\stackrel{\omerge}{\rightarrow}} 
\newcommand{\normf}[1]{#1\!\!\downarrow} 

\newcommand{\LK}{\ensuremath{\mathbf{LK}}}
\newcommand{\LKE}{\ensuremath{\mathbf{LKE}}}

\newcommand{\lkforall}[3][]{\infer[\forall{#1}]{#2}{#3}}
\newcommand{\lkexists}[3][]{\infer[\exists{#1}]{#2}{#3}}
\newcommand{\lkand}[3][]{\infer[\land{#1}]{#2}{#3}}
\newcommand{\lkor}[3][]{\infer[\lor{#1}]{#2}{#3}}

\newcommand{\lkcut}[3][]{\infer[\mathrm{cut}{#1}]{#2}{#3}}

\newcommand{\rk}{\ensuremath{\mathrm{rk}}}   
\newcommand{\cl}{\ensuremath{\mathrm{dl}}}   

\newcounter{tmprecount}
\newcommand{\thmrecount}[2]{\setcounter{tmprecount}{\value{theorem}}\setcounter{theorem}{#1}#2\setcounter{theorem}{\value{tmprecount}}} 

\title{Expansion Trees with Cut}

\author{Stefan Hetzl and Daniel Weller}

\date{April 15, 2013}

\begin{document}
\maketitle

\begin{abstract}
Herbrand's theorem is one of the most fundamental insights in logic. From
the syntactic point of view it suggests a compact representation of proofs
in classical first- and higher-order logic by recording the information
which instances have been chosen for which quantifiers, known in the
literature as expansion trees.

Such a representation is inherently analytic and hence corresponds to
a cut-free sequent calculus proof. Recently several extensions of
such proof representations to proofs with cut have been proposed. These
extensions are based on graphical formalisms similar to proof nets
and are limited to prenex formulas.

In this paper we present a new approach that directly extends expansion trees
by cuts and covers also non-prenex formulas. We describe a cut-elimination procedure
for our expansion trees with cut that is based on the natural reduction steps.
We prove that it is weakly normalizing using methods from the $\varepsilon$-calculus.
\end{abstract}

\section{Introduction}

Herbrand's theorem~\cite{Herbrand30Recherches,Buss95Herbrand}, one of the
most fundamental insights of logic, characterizes the
validity of a formula in classical first-order logic by the existence of a propositional
tautology composed of instances of that formula.

From the syntactic point of view this theorem induces a way of describing proofs:
by recording which instances have been picked for which quantifiers we
obtain a description of a proof up to its propositional part, a part we
often want to abstract from. An example for a formalism that carries out
this abstraction are Herbrand proofs~\cite{Buss95Herbrand}.
This generalizes nicely to most classical systems with quantifiers,
for example to simple type theory as in the expansion tree proofs of~\cite{Miller87Compact}.
Such formalisms are compact and useful proof certificates in many situations; they
are for example produced naturally by methods of automated deduction such
as instantiation-based reasoning~\cite{Korovin09Instantiation}.

These formalisms consider only instances of the formula that has been proved
and hence are {\em analytic} proof formalisms (corresponding to cut-free proofs
in the sequent calculus). Considering an expansion tree to be a compact
representation of a proof, it is thus natural to ask about the possibility of extending this kind of representation
to {\em non-analytic} proofs (corresponding to proofs with cut in the sequent
calculus).

In addition to enlarging the scope of instance-based proof representations, the
addition of cuts to expansion trees also sheds more light on the computational
content of classical logic. This is a central topic of proof theory and has
therefore attracted considerable attention, see ~\cite{Parigot92LambdaMu,Danos97New,Curien00Duality},
\cite{Barbanera96Symmetric},
\cite{Urban00Classical,Urban01Strong},
\cite{Berger02Refined},
\cite{Kohlenbach08Applied}, or~\cite{Baaz00CutElimination}, for
different investigations in this direction and~\cite{Avigad10Computational}
for a survey covering classical arithmetic.

Two instance-based proof formalisms incorporating a notion of cut have recently been proposed: proof
forests~\cite{Heijltjes10Classical} and Herbrand nets~\cite{McKinley13Proof}.
While proof forests are motivated by the game semantics for classical arithmetic
of~\cite{Coquand95Semantics}, Herbrand nets are based on methods for
proof nets~\cite{Girard87Linear}. These two formalisms share a number of
properties: both of them work in a graphical notation for proofs,
both work on prenex formulas only, for both weak but no strong normalization results are
known.

In this paper we present a new approach which works directly in the
formalism of expansion tree proofs and hence naturally extends the existing
literature in this tradition. As~\cite{Heijltjes10Classical,McKinley13Proof} we define
a cut-elimination procedure and prove it weakly normalizing but in contrast
to~\cite{Heijltjes10Classical,McKinley13Proof} we also treat non-prenex formulas,
therefore avoiding the distortion of the intuitive meaning of a formula
by prenexification.

We describe expansion trees with cuts for non-prenex end-sequents and cuts,
including their correctness criterion and how to translate from and to
sequent calculus. We describe natural cut-reduction steps and show that
they are weakly normalizing. A technical key for proving weak normalization
is to use methods of Hilbert's $\varepsilon$-calculus
which is a formalism for representing non-analytic first-order proofs
modulo propositional logic. The reader is invited to compare our treatment,
in particular the termination measure, with
the proof of the first $\varepsilon$-theorem in~\cite{Hilbert39Grundlagen2}, see~\cite{Moser06Epsilon}
for an exposition in English.

\section{Expansion Trees}

In this whole paper we work with classical first-order logic. Formulas and
terms are defined as usual. In order to simplify the exposition, we restrict
our attention to formulas in negation normal form (NNF). Mutatis mutandis all
notions and results of this paper generalize to arbitrary formulas. We write
$\dual{A}$ for the de Morgan dual of a formula $A$. A {\em literal} is an
atom $P(t_1,\ldots,t_n)$ or a negated atom $\dual{P}(t_1,\ldots,t_n)$.
\begin{definition}\label{def:exptree}
Expansion trees and a function $\Shallow(\cdot)$ (for {\em shallow}) that maps
an expansion tree to a formula are defined inductively as follows:
\begin{enumerate}
\item A literal $L$ is an expansion tree with $\Shallow(L)=L$.
\item If $E_1$ and $E_2$ are expansion trees and $\circ \in \{ \land,\lor \}$, then
$E_1\circ E_2$ is an expansion tree with $\Shallow(E_1\circ E_2) = \Shallow(E_1)\circ \Shallow(E_2)$.
\item If $\{ t_1,\ldots,t_n \}$ is a set of terms and $E_1,\ldots,E_n$ are
expansion trees with $\Shallow(E_i) = A\unsubst{x}{t_i}$
for $i=1,\ldots,n$, then
$E = \exists x\, A +^{t_1} E_1 \cdots +^{t_n} E_n$ is an expansion tree with
$\Shallow(E) = \exists x\, A$.
\item If $E_0$ is an expansion tree with $\Shallow(E_0)=A\unsubst{x}{y}$, then
$E = \forall x\, A +^y E_0$ is an expansion tree with $\Shallow(E) = \forall x\, A$.
\end{enumerate}
\end{definition}
The $+^{t_i}$ are called {\em $\exists$-expansions}
and the $+^\alpha$ {\em $\forall$-expansions},
and both $\forall$- and $\exists$-expansions
are called {\em expansions}.
The variable $y$ of a $\forall$-expansion $+^y$ is called
{\em eigenvariable} of this expansion. 
We say that $+^{t_i}$ {\em dominates} all the expansions
in $E_i$. Similarly, $+^\alpha$ {\em dominates} all the expansions
in $E_0$. 
\begin{definition}
We define the function $\Deep(\cdot)$ (for {\em deep}) that maps an expansion
tree to a formula as follows:
\begin{align*}
\Deep(L) &= L\ \mbox{for a literal}\ L,\\
\Deep(E_1\circ E_2) &= \Deep(E_1)\circ \Deep(E_2)\ \mbox{for $\circ\in\{ \land, \lor \}$},\\
\Deep(\exists x\, A +^{t_1} E_1 \cdots +^{t_n} E_n) &= \Lor_{i=1}^n \Deep(E_i),\ \mbox{and}\\
\Deep(\forall x\, A +^{y} E_0) &= \Deep(E_0).
\end{align*}
\end{definition}
We also say that $E$ is an {\em expansion tree of $\Shallow(E)$}.
\begin{definition}
A {\em cut} is a set $C = \{E_1, E_2\}$ of two expansion trees s.t.\ 
$\Shallow(E_1) = \dual{\Shallow(E_2)}$. 
\end{definition}
A formula is called positive if its top connective is $\lor$ or $\exists$ or
a positive literal. An expansion tree $E$ is called positive if $\Shallow(E)$ is
positive. It will sometimes be useful to consider a cut as an ordered pair: to that aim
we will write a cut as $C = (E_1,E_2)$ with parentheses instead of curly braces
with the convention that $E_1$ is the positive expansion tree. For a cut
$C = (E_1,E_2)$, we define $\Shallow(C) = \Shallow(E_1)$ which is also called
{\em cut-formula} of $C$. We define $\Deep(C) = \Deep(E_1) \land \Deep(E_2)$
\begin{definition}\label{def:exp_preproof}
Let $\calC$ be a set of cuts with pairwise different cut-formulas and let
$\calE$ be a set of expansion trees of pairwise different formulas. Then
$\calP = \calC, \calE$ is called {\em expansion pre-proof} if each two $\forall$-expansions
in $\calP$ have different eigenvariables ({\em regularity}), and if $\Shallow(\calP)$
does not contain free variables.
\end{definition}
For an expansion pre-proof $\calP = \calC,\calE$ we define $\Shallow(\calP) = \Shallow(\calE)$, which
corresponds to the end-sequent of a sequent calculus proof, and
$\Deep(\calP) = \Deep(\calE),\Deep(\calC)$ (which is a sequent of quantifier-free
formulas).
For an eigenvariable $\alpha$ in $\calP$,
define $q(\alpha)$ to be the $\forall$-expansion whose eigenvariable it is.
%
\begin{example}\label{ex.exp_preproof}
Consider the straightforward proof of $P(a)\impl \exists z\, Q(z)$ from
$\exists y\forall x\, (P(x) \impl Q(f(y)))$ via a cut on
$\forall x\exists y\, (P(x) \impl Q(f(y)))$. In negation normal formal
these formulas are $\dual{P}(a)\lor \exists z\, Q(z)$,
$\exists y\forall x\, (\dual{P}(x) \lor Q(f(y)))$,
and $\forall x\exists y\, (\dual{P}(x)\lor Q(f(y)))$. The proof will
be represented by the expansion pre-proof $\calP = \{ E^+, E^- \}, E_1, E_2$ where
\begin{align*}
E^+ = &\ \exists x\forall y\, ( P(x)\land \dual{Q}(f(y))) +^a (\ \forall y\, ( P(a) \land \dual{Q}(f(y)))
       +^\gamma P(a)\land\dual{Q}(f(\gamma))
        \ ) \\
E^- = &\ \forall x\exists y\, ( \dual{P}(x)\lor Q(f(y)) ) +^\beta (\ \exists y\, ( \dual{P}(\beta) \lor Q(f(y))) +^\alpha ( \dual{P}(\beta) \lor Q(f(\alpha)))
\ )\\
E_1 = &\ \forall y\exists x\, (P(x) \land \dual{Q}(f(y))) +^\alpha (\ \exists x\, (P(x)\land \dual{Q}(f(\alpha))) +^\beta P(\beta)\land\dual{Q}(f(\alpha)) 
\ )\\
E_2 = &\ \dual{P}(a) \lor ( \exists z\, Q(z) +^{f(\gamma)} Q(f(\gamma)) )
\end{align*}
We have $\Shallow(\calP) = \Shallow(E_1,E_2) =
\forall y\exists x\, (P(x) \land \dual{Q}(f(y))), \dual{P}(a) \lor \exists z\, Q(z)$
and
\begin{align*}
\Deep(\calP) =&\ \Deep(E^+)\land \Deep(E^-), \Deep(E_1), \Deep(E_2)\\
 =&\ (P(a) \land \dual{Q}(f(\gamma))) \land ( \dual{P}(\beta) \lor Q(f(\alpha))), P(\beta)\land\dual{Q}(f(\alpha)), \dual{P}(a)\lor Q(f(\gamma))
\end{align*}
\end{example}
As in~\cite{Heijltjes10Classical,McKinley13Proof} it would also be possible in our
setting to use a graphical notation.
However, we refrain from doing so in order to avoid
the parallel use of two different notations: a graphical for examples and
a more abstract notation for carrying out proofs.

Let us now move on to isolating the proofs in the set of pre-proofs.
The correctness criterion of expansion tree proofs~\cite{Miller87Compact}, but
also those of proof forests~\cite{Heijltjes10Classical} and Herbrand nets~\cite{McKinley13Proof}, has
two (main) components: 1.~a tautology-condition on one or more quantifier-free
formulas and 2.~an acyclicity condition on one or more orderings. While the tautology
condition of~\cite{Miller87Compact} generalizes to the setting of cuts in a straightforward
way, the acyclicity condition needs a bit more work: in the setting of cut-free
expansion trees it is enough to require the acyclicity of an order on the $\exists$-expansions.
%
In our setting that includes cuts we also have to speak about the order of
cuts (w.r.t.\ each other and w.r.t.\ $\exists$-expansions).
To simplify our treatment of this order we also include $\forall$-expansions.
Together this leads to the following
inference ordering constraints in expansion proofs.
\begin{definition}
Let $\calP = \calC, \calE$ be an expansion pre-proof. We will define the
{\em dependency relation} $<_\calP$, which is a binary relation on the set
of expansions and cuts in $\calP$. First,
we define the binary relation $<^0_\calP$ (writing $<^0$ if $\calP$ is clear
from the context) as the least relation satisfying ($C$ being a cut in $\calP$):
%
%
%
\begin{enumerate}
\item $v <^0 w$ if $w$ is an $\exists$-expansion in $\calP$ whose term contains the eigenvariable of the $\forall$-expansion $v$
\item $v <^0 w$ if $v$ is an expansion in $\calP$ that dominates the expansion $w$
\item $C <^0 v$ if $v$ is an expansion in $C$
\item $v <^0 C$ if $\Shallow(C)$ contains the eigenvariable of the $\forall$-expansion $v$
\end{enumerate}
$<_\calP$ is then defined to be the transitive closure of $<^0$. Again,
we write $<$ for $<_\calP$ if $\calP$ is 
clear from the context.
\end{definition}
%
%
\begin{definition}\label{def:exp_proof}
An {\em expansion proof} is an expansion pre-proof $\calP$ that satisfies the
following conditions:
\begin{enumerate}
\item $<_\calP$ is acyclic (i.e.~$x<_\calP x$ holds for no $x$),
\item $\Deep(\calP)$ is a tautology.
\end{enumerate}
\end{definition}
As there is no cycle containing cuts only, $<_\calP$ is cyclic iff $w<_\calP w$ for an
expansion $w$, and we will make use of this property without further mention.
%
%
\begin{example}
Coming back to the expansion pre-proof $\calP$ of Example~\ref{ex.exp_preproof}, note
that $\Deep(\calP) = (P(a) \land \dual{Q}(f(\gamma))) \land ( \dual{P}(\beta)
\lor Q(f(\alpha))), P(\beta)\land\dual{Q}(f(\alpha)), \dual{P}(a)\lor Q(f(\gamma))$
is a tautology (of the form $A\land B, \dual{B}, \dual{A}$). Let us now consider
the theory induced by $\calP$: in $\calP$ each term belongs to at most one $\exists$-
and at most one $\forall$-expansion
In such a situation we can uniformly notate all expansions as $Qt$ for some
term $t$ and $Q\in\{ \exists, \forall \}$. The expansions of $\calP$ are
then written as $\exists a$, $\forall \gamma$, $\forall \beta$, $\exists \alpha$,
$\forall \alpha$, $\exists \beta$, and $\exists f(\gamma)$. Furthermore,
$\calP$ contains a single cut $C$. Then $<^0$ is exactly:
\begin{enumerate}
 \item $\forall \gamma <^0 \exists f(\gamma)$, $\forall \beta <^0 \exists \beta$,
$\forall \alpha <^0 \exists \alpha$,
\item $\exists a <^0 \forall \gamma$, $\forall \beta <^0 \exists \alpha$, $\forall \alpha <^0 \exists \beta$,
\item $C <^0 \exists a$, $C <^0 \forall \gamma$, $C <^0 \forall \beta$, $C <^0 \exists \alpha$,
\item there is no $v<^0 C$ as the cut formula of $C$ is variable-free.
\end{enumerate}
As the reader is invited to verify, $<$ is acyclic.
\end{example}

\section{Basic Operations on Expansion Proofs}

Our cut-elimination algorithm, described in Section~\ref{sec:cutelim}, will
be based on natural rewrite rules of expansion proofs. In order to fully
specify those, we first need to clarify some basic operations on expansion
proofs.

\subsection{Expansion Trees with Merges}\label{subsec:merge}
One on these basic operations is the {\em merge} of expansion pre-proofs. If we
have two expansion pre-proofs $E_1$ and $E_2$ with $\Shallow(E_1) = \Shallow(E_2)$
we want to define a new expansion pre-proof $E_1\merge E_2$ which merges
$E_1$ and $E_2$. For example
\[
(\exists x\, P(x) +^{a} P(a)) \merge (\exists x\, P(x) +^{b} P(b)) =
\exists x\, P(x) +^{a} P(a) +^{b} P(b).
\]
In general however, this operation
can be considerably more complicated.

%
%


\begin{example}\label{ex.exptree_subst_merge}
Consider the following merge operation in an expansion pre-proof:
\begin{gather*}
( \forall x\, A +^u E_1 ) \merge ( \forall x\, A +^v E_2 ) , \exists x\, B +^{f(u)} F_1 +^{f(v)} F_2.
\end{gather*}
When propagating the merge node into the subtrees of the two trees being
merged, the two eigenvariables $u$ and $v$ will need to be unified, say by globally
applying the substitution $\unsubst{v}{u}$.
As eigenvariables are global, the result of this unification is that the two
$\exists$-expansions $+^{f(u)}$ and $+^{f(v)}$ in the expansion tree of $\exists x\, B$
will also be identified, violating the set-nature of the expansions of an
existential formula.
Globally applying the substitution $\unsubst{v}{u}$ therefore requires
merging the two trees $F_1\unsubst{v}{u}$ and $F_2\unsubst{v}{u}$.
\end{example}
We hence see that carrying out a merge operation does not only induce other
merge operations on subtrees but also substitutions and vice versa: carrying
out a substitution may induce additional merge operations. In order to give
a clear formal definition of these operations we will consider
{\em expansion pre-proofs with merges}: a data structure of expansion pre-proofs
which, in addition, contains an object-level merge-operation $\omerge$.
\begin{definition}
An \emph{expansion tree with merges} is defined by the same inductive
definition as expansion trees in Definition~\ref{def:exptree} to which we add the
following clause:
\begin{enumerate}\addtocounter{enumi}{4}
\item If $E_1$ and $E_2$ are expansion trees with merges s.t.\ $\Shallow(E_1) = \Shallow(E_2)$,
then $E_1\omerge E_2$ is an expansion tree with merges and
$\Shallow(E_1 \omerge E_2) = \Shallow(E_1) = \Shallow(E_2)$.
\end{enumerate}
We also extend $\Deep(\cdot)$ to expansion trees with merges by
setting $\Deep(E_1 \omerge E_2) = \Deep(E_1) \lor \Deep(E_2)$.
{\em Expansion (pre-)proofs with merges} are defined analogously
to expansion (pre-)proofs (without merge).
\end{definition}

\subsection{Substitution}
%
We now develop the definition of substitution via expansion trees with merges
indicated in the beginning of this section.
In the following, for a formula or term $F$ we denote by $\Var(F)$ the
set of variables free in $F$.
To make
sure that the application of a substitution transforms expansion trees (with
merges) into expansion trees (with merges) we have to restrict the set of permitted
substitutions: a substitution $\sigma$ can only be applied to an expansion tree
(with merges) $E$ if it is a renaming on the eigenvariables of $E$, more
precisely: if $\alpha\in\EV(E)$ implies that $\alpha\sigma$ is a variable.
Otherwise it would destroy the $\forall$-expansions.
Furthermore, to ensure no cycles are introduced in the dependency relation,
we have to impose
an additional restriction on the eigenvariables introduced by $\sigma$:
$\beta\in\Var(\alpha\sigma)$ implies that for all $\exists$-expansions $w$
in $\calP$ with an expansion term $t$ such that $\alpha\in \Var(t)$, we
have $w\not<q(\beta)$. 
A substitution fulfilling these conditions will be called {\em admissible
for $\calP$}.

Later we will give an operational meaning to the merge by means of a reduction system. 
This will allow us to define a notion of substitution for expansion tree proofs
{\em without merge}.
%
\begin{definition}\label{def:subst}
Let $E$ be an expansion tree with merges and let $\sigma$ be a substitution.
\begin{enumerate}
\item For a literal $L$, $L\sigma$ is defined as for formulas.
\item $(E_1 \circ E_2)\sigma = E_1\sigma \circ E_2\sigma$ for $\circ \in \{ \land, \lor \}$.
\item Let $E = \exists x\, A +^{t_1} E_1 \cdots +^{t_n} E_n$, let $\{ s_1,\ldots, s_k \}$
be $\{ t_1\sigma,\ldots, t_n\sigma\}$ and define
\[
  E\sigma =
\exists x\, A\sigma +^{s_1} \oMerge_{\scriptsize\begin{array}{c}i\in\{1,\ldots,n\}\\ t_i\sigma = s_1\end{array}}
E_i\sigma \cdots +^{s_k} \oMerge_{\scriptsize\begin{array}{c}i\in\{1,\ldots,n\}\\ t_i\sigma = s_k\end{array}}
E_i\sigma.\]
\item $(\forall x\, A +^\alpha E)\sigma = \forall x\, A\sigma +^{\alpha\sigma} E\sigma$.
\item $(E_1 \omerge E_2)\sigma = E_1\sigma \omerge E_2\sigma$.
\end{enumerate}
For an expansion pre-proof $\calP = C_1,\ldots, C_k, E_1,\ldots, E_n$ and
a substitution $\sigma$ s.t.\ $\alpha\in\EV(\calP)$ implies that $\alpha\sigma$
is a variable we define $\calP\sigma = C_1\sigma,\ldots, C_k\sigma, E_1\sigma,\ldots, E_n\sigma$.
\end{definition}
%
To every expansion $w$ in $\calP\sigma$ we can naturally associate a non-empty set of predecessors w.r.t.\ substitution 
$\pred_s(w)$ in $\calP$ (note that $\pred_s(w)$ is always a singleton, except
in case 3 of the above definition).
%
As usual in the term
rewriting literature, $\calP[]$ denotes an expansion pre-proof context,
i.e.\ an expansion pre-proof with a hole and $\calP[E]$ denotes the
expansion pre-proof obtained from filling this hole with the expansion
tree $E$.
\begin{lemma}\label{lem:subst_proof}
Let $\calP=\calP'[E]$ be an expansion proof with merges and $\sigma$
a substitution admissible for $\calP$. Then $\calQ=\calP'[E\sigma]$ 
is an expansion proof with merges, and $\Shallow(\calP)=\Shallow(\calQ)$.
\end{lemma}
\begin{proof}
The existence of a cycle in $\calQ$ implies that of one in $\calP$,
see Appendix for details.
\end{proof}

\subsection{Merge}
As we have seen in Example~\ref{ex.exptree_subst_merge}, carrying out a merge
operation may require to identify two eigenvariables globally, i.e.\ on the
level of the expansion pre-proof. The object-level merge operations are hence
executed by the following reduction system which, in addition to local
term rewriting, includes global variable renaming. 
\begin{definition}\label{def:omerge_reduction}
We define a reduction system $\stackrel{\omerge}{\mapsto}$
on expansion pre-proofs with merges.
\begin{enumerate}
\item $\calP[L \omerge L] \stackrel{\omerge}{\mapsto} \calP[L]$ for a literal $L$. 

\item $\calP[(E'_1\circ E''_1)\omerge (E'_2\circ E''_2)] \stackrel{\omerge}{\mapsto}
\calP[(E'_1\omerge E'_2)\circ(E''_1\omerge E''_2)]$ for $\circ \in \{ \land, \lor \}$.

\item $\calP[(\forall x\, A +^{\alpha_1} E_1)\omerge (\forall x\, A +^{\alpha_2} E_2)]
\stackrel{\omerge}{\mapsto} \calP[\forall x\, A +^{\alpha_1} (E_1 \omerge E_2)]\unsubst{\alpha_2}{\alpha_1}$.

\item If $E_1 = \exists x\, A +^{r_1} E_{1,1} \ldots +^{r_k} E_{1,k} +^{s_1}
F_1\ldots +^{s_l} F_l$ and
$E_2 = \exists x\, A +^{r_1} E_{2,1}\ldots +^{r_k} E_{2,k} +^{t_1} G_1\ldots
+^{t_m} G_m$ where $\{s_1,\ldots,s_l\}\cap\{t_1,\ldots,t_m\}=\emptyset$, then
\[
 \calP[E_1\omerge E_2]\stackrel{\omerge}{\mapsto}
 \calP[\exists x\, A +^{r_1} (E_{1,1}\omerge E_{2,1}) \ldots +^{r_k} (E_{1,k}\omerge E_{2,k})
 +^{s_1} F_1\ldots +^{s_l} F_l
 +^{t_1} G_1\ldots +^{t_m} G_m]
\]
\end{enumerate}
Write $\mergered$ for the reflexive and transitive closure of $\stackrel{\omerge}{\mapsto}$.
\end{definition}
As with substitution, for $\calP \stackrel{\omerge}{\mapsto} \calP'$
we associate to every expansion $n$ in $\calP'$ a non-empty set $\pred^0_\omerge(n)$
of predecessor expansions from $\calP$ in the natural way, noting that
$\pred^0_\omerge(n)$ is a singleton in cases $1,2$ of the definition,
and contains at most 2 elements in cases $3,4$.
We extend $\pred^0_\omerge$ to $\mergered$ by denoting the 
reflexive and transitive closure of $\pred^0_\omerge$ by
$\pred_\omerge$.

%
%
%
\begin{lemma}\label{lem:mergered_props}
The relation $\mergered$ is confluent and strongly
normalizing. Its normal forms have no merge nodes.
\end{lemma}
\begin{proof}
See Appendix.
\end{proof}
By $\normf{\calP}$ we denote the normal form of $\calP$ under $\mergered$.
%
%
%
%
%
We now use the above reduction system on object-level merge nodes
for defining the actual merge operation on expansion trees without merge nodes.
\begin{definition}
Let $E_1,E_2$ be expansion trees with $\Shallow(E_1) = \Shallow(E_2)$, then
$E_1 \merge E_2$ is defined as $\normf{(E_1 \omerge E_2)}$.
\end{definition}
The merge operation is extended to expansion pre-proofs in the natural way:
expansion trees and cuts with the same shallow formula are merged, the others
are combined by set-theoretic union, where merging of a cut is defined as follows:
%
for cuts $C_1 = ( E^+_1, E^-_1 )$ and $C_2 = ( E^+_2, E^-_2 )$ with
$\Shallow(C_1) = \Shallow(C_2)$ we define $C_1 \merge C_2$ as
$( E^+_1 \merge E^+_2, E^-_1 \merge E^-_2 )$.
%
%
%
\begin{lemma}\label{lem:merge_proof}
  If $\calP_1\omerge\calP_2$ is an expansion proof with merge such that $\Shallow(\calP_1)=\Shallow(\calP_2)$,
  then $\calP_1\cup \calP_2$ is an
  expansion proof and $\Shallow(\calP_1\cup \calP_2)=\Shallow(\calP_1)=\Shallow(\calP_2)$.
\end{lemma}
\begin{proof}
See Appendix.
\end{proof}
%


The role of the merge operation is to recursively identify such variables that
denote the same value. For the purpose of cut-elimination, its principal
use consists in defining which parts of an expansion tree are to be
duplicated by a reduction. It is not surprising that this is technically
involved as it is also the case in other comparable formalisms. Indeed, it
is maybe in the technical details of {\em how} the decision what to duplicate
is taken where the existing formalisms differ most: in the
$\varepsilon$-calculus~\cite{Hilbert39Grundlagen2}, the object-level
syntax of $\varepsilon$-terms ensures maximal identifications, in proof
forests~\cite{Heijltjes10Classical}, the reduction steps duplicate too much
and are hence interleaved with pruning steps and in
Herbrand nets~\cite{McKinley13Proof} the notion of kingdom from the
literature on proof nets is used for determining what to duplicate.

\section{Expansion Proofs and Sequent Calculus}

In this section we will clarify the relationship between our expansion proofs
and the sequent calculus. The concrete version of sequent calculus is of
no significance to the results presented here, they
hold mutatis mutandis for every version that is common in the literature.
For technical convenience
we choose a calculus where a sequent is a set of formulas and
all rules are invertible.
\begin{definition}
The calculus {\LK} is defined as follows: initial sequents are of the form
$\Gamma, A, \dual{A}$ for an atom $A$. The inference rules are
\[
\lkforall{\Gamma, \forall x\, A}{\Gamma, A\unsubst{x}{\alpha}}
\quad 
\lkexists{\Gamma, \exists x\, A}{\Gamma, \exists x\, A, A\unsubst{x}{t}}
\quad
\lkand{\Gamma, A\land B}{\Gamma, A & \Gamma, B}
\quad
\lkor{\Gamma, A\lor B}{\Gamma, A, B}
\quad
\lkcut{\Gamma}{\Gamma, A & \dual{A}, \Gamma}
\]
with the usual side conditions: $\alpha$ must not appear in $\Gamma, \forall x\, A$
and $t$ must not contain a variable which is bound in $A$.
\end{definition}
Due to the global nature of expansion proofs, they correspond to regular
{\LK}-proofs.
An {\LK}-proof is called {\em regular} if each two $\forall$-inferences have
different eigenvariables.
From now on we assume w.l.o.g.\ that all {\LK}-proofs are regular.

\subsection{From Sequent Calculus to Expansion Proofs}
In this section we describe how to read off expansion trees from {\LK}-proofs
which leads to a completeness theorem for expansion proofs. For representing
a formula $A$ that is introduced by (implicit) weakening we use the natural
coercion of $A$ into an expansion tree, denoted by $A^\mathrm{E}$.
%
For a sequent $\Gamma = A_1,\ldots,A_n$ we define
$\Gamma^\mathrm{E} = A_1^\mathrm{E},\ldots, A_n^\mathrm{E}$.

\begin{definition}
For an {\LK}-proof $\pi$ define the expansion proof $\Exp(\pi)$ by
induction on $\pi$:
\begin{enumerate}
\item If $\pi$ is an initial sequent $\Gamma, A, \dual{A}$, then
$\Exp(\pi) = \Gamma^\mathrm{E}, A, \dual{A}$

\item If 
$
\pi
=
\begin{array}{c}
\lkand{\Gamma, A \land B}{
  \deduce{\Gamma, A}{(\pi_A)}
  &
  \deduce{\Gamma, B}{(\pi_B)}
}
\end{array}
$
with $\Exp(\pi_A) = \calP_A, E_A$ and $\Exp(\pi_B) = \calP_B, E_B$ where
$\Shallow(E_A) = A$ and $\Shallow(E_B) = B$, then
$\Exp(\pi) = \calP_A \merge \calP_B, E_A \land E_B$.

\item If
$
\pi
=
\begin{array}{c}
\lkor{\Gamma, A\lor B}{
  \deduce{\Gamma, A, B}{(\pi')}
}
\end{array}
$
with $\Exp(\pi') = \calP, E_A, E_B$ where $\Shallow(E_A) = A$ and $\Shallow(E_B) = B$,
then $\Exp(\pi) = \calP, E_A \lor E_B$.

\item If
$
\pi
=
\begin{array}{c}
\lkforall{\Gamma, \forall x\, A}{
  \deduce{\Gamma, A\unsubst{x}{\alpha}}{(\pi_A)}
}
\end{array}
$
with $\Exp(\pi_A) = \calP, E$ where $\Shallow(E) = A\unsubst{x}{\alpha}$,
then $\Exp(\pi) = \calP, \forall x\, A +^\alpha E$.

\item If
$
\pi
=
\begin{array}{c}
\lkexists{\Gamma, \exists x\, A}{
  \deduce{\Gamma, \exists x\, A, A\unsubst{x}{t}}{(\pi_A)}
}
\end{array}
$
with $\Exp(\pi_A) = \calP, E, E_t$ where $\Shallow(E) = \exists x\, A$ and
$\Shallow(E_t) = A\unsubst{x}{t}$, then $\Exp(\pi) = \calP,
E\merge \exists x\, A +^t E_t$.

\item If
$
\pi
=
\begin{array}{c}
\lkcut{\Gamma}{
  \deduce{\Gamma, A}{(\pi^+)}
  &
  \deduce{\dual{A}, \Gamma}{(\pi^-)}
}
\end{array}
$
for $A$ positive with $\Exp(\pi^+) = \calP^+, E^+$ and $\Exp(\pi^-) = \calP^-, E^-$
where $\Shallow(E^+) = A$ and $\Shallow(E^-) = \dual{A}$, then $\Exp(\pi) = (E^+, E^-), \calP^+ \merge \calP^-$.
\end{enumerate}
\end{definition}

Note that the behavior of the above definition of $\Exp(\cdot)$ on binary rules
is to merge expansions of both subproofs (including cuts). This is the
reason for the relationship between sequent calculus proofs and expansion
proofs which on the one hand are strongly connected
structurally~\cite{Chaudhuri12Systematic,ChaudhuriXXIsomorphism} but at
the same time have different complexity~\cite{Baaz12Complexity}.
\begin{theorem}[completeness]\label{thm:completeness}
If $\pi$ is an {\LK}-proof of a sequent $\Gamma$, then $\Exp(\pi)$ is an
expansion proof of $\Gamma$. If $\pi$ is cut-free then so is $\Exp(\pi)$.
\end{theorem}%
\begin{proof}
That $\Exp(\pi)$ is an expansion pre-proof follows directly from the
definitions as we are dealing with regular {\LK}-proofs only.
By a straightforward induction on $\pi$ one shows that $\Deep(\Exp(\pi))$
is a tautology.
Acyclicity is also shown inductively by observing that
if $\alpha$ is a free variable in the end-sequent of $\pi$,
then $\alpha$ is not an eigenvariable in $\Exp(\pi)$. This
implies that if $w$ is the new expansion introduced in
the construction of $\Exp(\pi)$, and $v$ is an old expansion
in $\Exp(\pi)$, then $w\not>v$, which in turn yields
acyclicity.
\end{proof}
\subsection{From Expansion Proofs to Sequent Calculus}
In this section we show how to construct an {\LK}-proof from a given expansion
proof. To this aim we introduce a calculus {\LKE} that works on expansion
pre-proofs instead of sequents (of formulas) following the treatment
in~\cite{Miller87Compact}.
\begin{definition}
The axioms of {\LKE} are of the form $\calP, A, \dual{A}$ for an atom $A$.
The inference rules are
\[
\lkforall{\calP, \forall x\, A +^\alpha E_0}{\calP, E_0}
\quad
\lkexists{\calP, \exists x\, A +^{t_1} E_1 \cdots +^{t_n} E_n}{
  \calP, \exists x\, A +^{t_1} E_1 \cdots +^{t_{n-1}} E_{n-1}, E_n
}
\]
\[
\lkand{\calP, E_1\land E_2}{\calP, E_1 & \calP, E_2}
\quad
\lkor{\calP, E_1\lor E_2}{\calP, E_1, E_2}
\quad
\lkcut{\{E_1,E_2\},\calP}{\calP, E_1 & E_2, \calP}
\]
with the following side conditions: $\Shallow(E_1) = \dual{\Shallow(E_2)}$ for the cut
and the eigenvariable condition for $\forall$: $\alpha$ must not occur in
$\Shallow(\calP, \forall x\, A +^x E_0)$.
\end{definition}
The reader is invited to note that $\Shallow(\calP, \forall x\, A +^x E_0)$ does
{\em not} include the cut formulas of $\calP$, they may -- and indeed often
have to -- contain the eigenvariable
$\alpha$. Furthermore, it should be kept in mind that the expansion terms at
the $\exists$-rule form a {\em set}, i.e.\ the above rule allows to take any
instance as there is no such thing as a last or rightmost instance. An important
feature of the above calculus, which is easily verified, is that if $\pi$ is an {\LKE}-proof,
then $\Shallow(\pi)$ is an {\LK}-proof. In the following proof we describe how
to transform expansion proofs to {\LK}-proofs.
\begin{theorem}[soundness]\label{thm:soundness}
If $\calP$ is an expansion proof of a sequent $\Gamma$, then there
is an {\LK}-proof of $\Gamma$. If $\calP$ is cut-free, then so is the {\LK}-proof.
\end{theorem}
\begin{proof}
It is enough to construct an {\LKE}-proof $\pi$ of $\calP$, as then $\Shallow(\pi)$ is
a proof of $\Shallow(\calP) = \Gamma$. The construction will be carried out
by induction on the number of nodes in $\calP$.

If $\calP = \calP', E_1 \lor E_2$ for some $\calP'$, $E_1$ and $E_2$,
then both $\calP', E_1, E_2$ is a strictly smaller expansion proof.
By the induction hypothesis we obtain an {\LKE}-proofs $\pi'$ of $\calP', E_1, E_2$
from which a proof of $\calP$ is obtained by an $\lor$-inference.
For $\calP = \calP', E_1 \land E_2$, proceed analogously.

If there are no top-level conjunctions or disjunctions, then by the acyclicity of
$<_\calP$ there must be a $<_\calP$-minimal top-level quantifier or cut. For
the case of cut proceed as follows: let $\calP = C,\calP'$ for some $\calP'$ and a $<_\calP$-minimal cut $C=\{ E_1, E_2 \}$.
Then both $E_1, \calP'$ and $E_2,\calP'$
are strictly smaller expansion proofs because $\Deep(E_i, \calP')$ is a tautology as $\Deep(\calP)$
is one and the orderings are suborderings of $\calP$ hence also acyclic. By
the induction hypothesis we obtain {\LKE}-proofs $\pi_1,\pi_2$ of $E_1, \calP'$
and $E_2,\calP'$ respectively from which a proof of $\calP$ is obtained
by a cut.

For the case of the minimal node being a quantifier, proceed analogously. As
in the cut-free case the eigenvariable condition of the $\forall$-rule is
ensured by the acyclicity of the dependency relation.
\end{proof}
\begin{definition}
 The {\LK}-proof constructed in the above proof will be called $\Seq(\calP)$.
\end{definition}

\section{Cut-Elimination}\label{sec:cutelim}
In this section we define a natural reduction system for expansion
proofs whose normal forms are cut-free expansion proofs.
We prove weak normalization and discuss the status of other properties such
as strong normalization and confluence in comparison to other systems from the literature.
\subsection{Cut-Reduction Steps}\label{sec:cut_reduction_steps}
%
Before we present our cut-reduction steps, we have to discuss regularity:
in contrast to the operations we have defined so far, cut-reduction will {\em duplicate}
sub-proofs, making it necessary to discuss the renaming of variables (as
in the case of the sequent calculus). 
We will carefully indicate,
in the case of a duplication, which subtrees should be subjected to a
variable renaming, and which variables are to be renamed.

%
%
The cut-reduction steps, relating expansion proofs $\calP,\calP'$ and
written $\calP\mapsto\calP'$,
are
\[
  \begin{array}{rl}
  & \{  \exists x\, A +^{t_1} E_1 \cdots +^{t_n} E_n , \forall x\, \bar{A} +^\alpha E \}, \calP\\
\mapsto& \calP \cup \{ E_1\lor\cdots\lor E_n, E\eta_1\unsubst{\alpha}{t_1}\land\cdots\land E\eta_n\unsubst{\alpha}{t_n} \} \cup \bigcup_{i=1}^n \calP\eta_i\unsubst{\alpha}{t_i}
  \end{array}
\]
\[
  \{ E_1 \lor E_2, E'_1\land E'_2 \}, \calP \mapsto \{ E_1, E'_1 \}\cup \{ E_2, E'_2 \}\cup \calP
\]
\[
\{ A, \dual{A} \}, \calP\mapsto \calP \quad \mbox{for an atom $A$.}
\]
where $\eta_i$ are renamings of the eigenvariables of $\calP,E$ to fresh variables.

These reduction rules are very natural:
an atomic cut is simply removed and a propositional cut is decomposed. 
The reduction of a quantified cut is, when thinking about cut-elimination in the
sequent calculus, intuitively immediately appealing: An existential cut is replaced
by a cut on a disjunction of the instances. We emphasize here that due to the eigenvariable condition
in the sequent calculus, such a rule cannot directly be stated with such formal
clarity and elegance.
Note that the rule makes use
of the merge operation which, as will become clear in the following sections, will prevent
redundancies that would be introduced by using the set-union $\cup$.

One surprising aspect of the quantifier-reduction rule is the presence of $\calP$, without
a substitution applied, on the rhs of the rule: in general, $\calP$ will contain $\alpha$, and
one would expect that occurrences of $\alpha$ are redundant (since $\alpha$ is ``eliminated''
by the rule). The reason why this occurrence of $\calP$ must be present is that $\alpha$ is not,
in fact, eliminated since some $t_i$ might contain it. This situation occurs, for example, when
translating from a regular $\LK$-proof where an $\exists$-quantifier may be instantiated by any
term, and we happen to choose an eigenvariable from a different branch of the proof. In the 
sequent calculus, this situation can in principle be avoided by using a different witness for the
$\exists$-quantifier, but realizing such a renaming in expansion proofs is technically non-trivial
due to the global nature of eigenvariables. For simplicity of exposition, we therefore allow this
somewhat unnatural situation and leave a more detailed analysis for future work.
\begin{remark}
We note that this phenomenon also occurs in the proof forests of~\cite{Heijltjes10Classical},
where it is called {\em bridge}. There,
bridges are dealt with by a {\em pruning reduction}, and the weak normalization proof of
that system depends on this pruning. In our setting, we do not need additional machinery
for proving weak normalization (see Section~\ref{sec:weak_normalization}). Furthermore,
the counterexample to strong normalization from~\cite{Heijltjes10Classical} also contains
a bridge; we investigate (a translation of) this counterexample in
Section~\ref{sec:strong_normalization} and find that it is not a counterexample for
our reduction.
\end{remark}
As before, if $\calP \mapsto \calP'$ we can associate in a natural way
(formally, using the $\pred_\omerge$ and $\pred_s$ functions defined before) to every expansion
$w$ in $\calP'$ a unique predecessor (w.r.t.\ cut-reduction) in $\calP$. This predecessor
is denoted by $\pred_c(w)$.
Note that $\pred_c(w)$ is a single expansion, while $\pred_\omerge(w)$ is a set of expansions;
this is explained by the fact that all expansions in $\pred_\omerge(w)$ are
,,copies'' of $\pred_c(w)$.

\begin{lemma}\label{lem:red_pres_proof}
If $\calP\mapsto\calP'$ and $\calP$ is an expansion proof, then $\calP'$ is an expansion
proof. Furthermore, $\Shallow(\calP)=\Shallow(\calP')$.
\end{lemma}
\begin{proof}
See Appendix.
\end{proof}
\begin{example}
For the sake of conciseness, we use the notation $E(\alpha)$ for an
expansion tree with an indicated variable $\alpha$, and $E(t)$ for
the expansion tree obtained from $E(\alpha)$ by (syntactically)
substituting $t$ for $\alpha$. We will also identify formulas and
quantifier-node-free expansion trees. With this in mind, consider
the expansion proof $\calP=\neg P0, Pf^40, E(\alpha), \{C^+,C^-\}$
with
\[
\begin{array}{lll}
E(\alpha)&= \exists x F(x) &+^\alpha F(\alpha)+^{f\alpha} F(f\alpha)\\
C^+&=\exists x G(x)&+^0 G(0)+^{f^20}G(f^20)\\
C^-&=\forall x\dual{G(x)}&+^\alpha \dual{G(\alpha)},
\end{array}
\]
where $F(x)=Px \land \neg Pfx$ and $G(x)=Px\land \neg Pf^2x$.
Then, since in this case substitution does not introduce any
merge nodes and no eigenvariable renaming is necessary,
\[
\calP\mapsto\neg P0,Pf^40,E(\alpha)\cup E(0)\cup E(f^20),\{
G(0) \lor G(f^20), \dual{G(0)}\land\dual{G(f^20})\}
\]
where the substitutions $\unsubst{\alpha}{0}, \unsubst{\alpha}{f^2 0}$ were applied and
\[
  E(\alpha)\cup E(0)\cup E(f^20)=\exists x F(x) +^\alpha F(\alpha)+^{f\alpha} F(f\alpha)+^0 F(0) +^{f0} F(f0)
  +^{f^2 0}F(f^20)+^{f^3 0}F(f^30).
\]
Finally, this proof reduces to
\[
\neg P0,Pf^40,E(\alpha)\cup E(0)\cup E(f^20)
\]
by the propositional cut-reduction rules. The reader is invited
to verify that tautology-hood of $\Deep(\calP)$ is preserved (the $\alpha$-instances
are redundant in this case). The final expansion proof does not contain any $\forall$-nodes,
so acyclicity of the dependency relation is trivial.
\end{example}
In the sequel, by $\rightarrow$ we denote the reflexive, transitive closure of the mapping $\mapsto$.
%
%
%
%

\subsection{Complexity Measures}
Our next aim is to prove weak normalization of our reduction
system $\rightarrow$. It turns out that the strategy of the proof of the {\em first
 $\varepsilon$-theorem} can be applied to expansion trees. For simplicity,
 we just state the second $\varepsilon$-theorem, which is a consequence of the first:
for every proof of an $\varepsilon$-free formula
in the $\varepsilon$-calculus, there exists a proof of the same 
formula in which no $\varepsilon$'s occur. It is known that proofs in the $\varepsilon$-calculus
can be translated to {\LK}-proofs with cut, and vice-versa. This translation
shows us that closed $\varepsilon$-terms correspond to eigenvariables in the sequent
calculus, which in turn correspond to $\forall$-expansions in expansion proofs.
Equipped with this observation, we can find suitable versions of the notions
of {\em rank} and {\em degree} which in turn will allow us to prove
weak normalization. In fact, these notions can be formulated in a natural way using the
language of expansion trees we have introduced so far.
In the following, we fix $\max\emptyset=0$.
\begin{definition}
 Let $w$ be a $\forall$-expansion in 
 $\calP$,
 and let $>$ be its dependency relation. A sequence of $\forall$-expansions
 $w,w_1,\ldots,w_k$ of $\calP$ such that $w>w_1>\cdots>w_k$ is called
 a {\em $>$-chain descending from $w$ of length $k$}.
%
%
%
%
%
%
%
%
%
%
%
%
%
 We now define the {\em rank $\rk(w)$} for expansions $w$
 and the {\em degree $\deg(w)$} for $\forall$-expansions $w$:
  \begin{align*}
    \rk(w)=&\max\{\rk(u)\mid w\textrm{ dominates } u\} + 1,\\
    \deg(w)=&\max\{\textrm{length of } c \mid c \textrm{ >-chain descending from }w\}.
  \end{align*}
\end{definition}
A trivial but crucial property of
$\deg$ is that it is order-preserving w.r.t.~the dependency relation,
i.e.~$v>w$ implies $\deg(v)>\deg(w)$. For use in our weak normalization
proof, we extend the notion of rank to expansion proofs,
calling expansions $w$ occurring in a cut {\em critical}.
\begin{definition}
  For an expansion proof $\calP$ and $r\in\nats$, 
  the {\em rank} $\rk(\calP)$ and the {\em order with respect to $r$} $o(\calP,r)$ are defined as
  \begin{align*}
    \rk(\calP)=&\max\{\rk(w)\mid w\textrm{ critical}\},\\
    o(\calP,r)=&\#\{w\mid w\textrm{ critical }\forall\textrm{-expansion} \land \rk(w)=r\}.
  \end{align*}
\end{definition}
\subsection{Elimination of Propositional Connectives}
Since expansion proofs work modulo propositional validity, it can be
expected that the elimination of propositional parts of cuts is simple.
This is indeed the case: from our cut-reduction steps, it is immediately clear that purely propositional cuts can be eliminated in linear time (since each propositional connective and each atom in a cut-formula induces a single
cut-reduction step). In fact, it is easy to see that if $\calP,\calC$ is an expansion proof where $\calC$ contains only propositional cuts, then $\calP$ is also an expansion proof. Hence purely propositional cuts can simply be dropped. This is in line with the results
of~\cite{Weller11Elimination}, where it is shown that quantifier-free
cuts can be eliminated from {\LK}-proofs at the cost of propositional proof search.

The following result builds on these observations, showing that
propositional parts of cuts can be eliminated while preserving the complexity measures we have defined in the previous section. This will
yield a convenient
,,intermediate
normal form'' that will be used in the proof of weak normalization.
\begin{definition}
  An expansion proof $\calP=\calC,\calE$ is {\em $\lor\land$-normal} if no $(E_1,E_2)\in\calC$
  is of the form $E_1=E_l\lor E_r$.
\end{definition}
In particular, if $\calP$ is $\lor\land$-normal and no cut
in $\calP$ contains a quantifier, then $\calP$ contains 
only atomic cuts.
%
%
\begin{lemma}\label{lem:lorland_normal}
  For every expansion proof $\calP$ there is a $\lor\land$-normal expansion proof $\calP^*$
  such that $\calP \rightarrow \calP^*$, $\Shallow(\calP^*) = \Shallow(\calP)$,
  $\rk(\calP^*)=\rk(\calP)$
  and 
  $o(\calP^*,r)=o(\calP,r)$ for all $r$.
\end{lemma}
\begin{proof}
We proceed by induction on the number of $\lor\land$-cuts in $\calP$,
showing by induction on the structure of $\calP$ that
$\rk$ is preserved. See the appendix for details.
\end{proof}

%
\subsection{Weak Normalization}\label{sec:weak_normalization}
%
%
This section is dedicated to proving that there exists a terminating
strategy for the application of the cut-reduction rules.
Given an expansion proof $\calP$, our reduction strategy will be based 
on picking a degree-maximal $\forall$-expansion from the set 
$M(\calP)=\{w\mid w\textrm{ critical and }\rk(w)=\rk(\calP)\}$. 
%
%
%
%
The following results establish some invariances of rank
and domination under
substitution and $\mergered$-reduction, which are crucial
for the weak normalization proof.
\begin{lemma}\label{lem:subst_rank_preserve}
  Let $v,w$ be expansions in $\normf{\calP[\alpha\backslash t]}$ such that
  $\alpha$ does not occur in any cut-formula in $\calP$, and let
  $v'\in\pred_\omerge(v)$ and $w'\in\pred_\omerge(w)$. Then $v$ dominates $w$ if and only if 
  $v'$ dominates $w'$. Furthermore, $\rk(v')=\rk(v)$.
\end{lemma}
\begin{proof}
By induction on the definition of $P\sigma$ and
$\calP_1 \stackrel{\omerge}{\mapsto} \calP_2$.
See the appendix for details.
\end{proof}
%
\begin{lemma}\label{lem:merge_deg}
  Let $\calP \mergered \calP'$ and $v$ be an expansion
in $\calP'$ and $v'\in\pred_\omerge(v)$. Then  
$\rk(v')=\rk(v)$.
\end{lemma}
\begin{proof}
By induction on a $\mergered$-sequence of $\calP$. 
\end{proof}
%
\begin{lemma}\label{lem:order_subst_merge}
Let $r=\rk(\calP)$ and $\sigma$ a substitution. Then $o(\calP\sigma,r)=o(\calP,r)$. 
Furthermore, let $E_1,E_2$ be expansion trees and
$E=\normf{(E_1\omerge E_2)}$,
then $o(E,r)= o(E_1,r)=o(E_2,r)$.
\end{lemma}
\begin{proof}
Using Lemma~\ref{lem:subst_rank_preserve} for substitution, and for
merge the fact that expansions of rank $r$ are uppermost
and hence merged. See Appendix for details.
\end{proof}
%
We are ready to state the main tool of the termination proof. 
It shows that when reducing an appropriate quantified cut, the
number of $\forall$-expansions with maximal rank decreases, while 
the maximal rank does not increase. The difficulty lies in showing that
while the expansion proof with merge that is constructed by the
cut-reduction rule may, in fact, contain more $\forall$-expansions of
maximal rank, this increase will be eliminated by the merge-normalization.
\begin{lemma}\label{lem:reduction:new}
Let $\calP_1\mapsto\calP_2$ by the quantifier reduction-rule, 
let $r=\rk(\calP_1)$ and denote
the reduced $\forall$-expansion by $w$. 
If $w\in M(\calP)$ and $\deg(w)$ is maximal in $M(\calP_1)$, then
$\rk(\calP_2)\leq r$
and $o(\calP_2,r)= o(\calP_1,r)-1$.
\end{lemma}
\begin{proof}
We have $\rk(\calP_2)\leq r$ since the rank changes for no expansions by
Lemmas~\ref{lem:subst_rank_preserve}~and~\ref{lem:merge_deg}.
To show that $o(\calP_2,r)= o(\calP_1,r)-1$, it suffices to show that
for all non-reduced cuts $G\in\calP_1$ containing a $\forall$-expansion of rank $r$,
$\alpha\notin\Var(\Shallow(G))$ and if $\beta\in\Var(\Shallow(G))$ then
$w\not<q(\beta)$: If this is so, then $\Shallow(G)\unsubst{\alpha}{t_i}=\Shallow(G)$ 
and 
(as can be checked by induction)
the regularization $\eta_i$ is reversed w.r.t.~$\Shallow(G)$
by the merge, and hence the cuts are merged. Therefore, by Lemma~\ref{lem:order_subst_merge},
their order stays the same. Furthermore, $E_i,E$ do not contain expansions of maximal
rank (since $w$ has maximal rank), and $w$ does not have a successor in $\calP_2$,
hence $o(\calP_2,r)= o(\calP_1,r)-1$.

To show the claim, consider a $\forall$-expansion $v$ of rank $r$ in $G$. 
If 
$\alpha\in\Var(\Shallow(G))$, then $v>q(\alpha)=w$ and
therefore $\deg(v)>\deg(w)$, which contradicts maximality of $\deg(w)$. 
Similarly, $\beta\in\Var(\Shallow(G))$ implies
$q(\beta)<w$. Assuming $w<q(\beta)$ yields $w<v$ and again the contradictory
$\deg(w)<\deg(v)$.
\end{proof}

\begin{theorem}[Weak Normalization]\label{thm:weak_normalization}
For every expansion proof $\calP$ there is a cut-free expansion proof $\calP^*$
with $\Shallow(\calP) = \Shallow(\calP^*)$ and $\calP \rightarrow \calP^*$.
\end{theorem}
\begin{proof}
First, we apply the propositional cut-reduction
rules exhaustively to $\calP$ to obtain an $\lor\land$-normal expansion proof $\calP^*$
(Lemma~\ref{lem:lorland_normal}). If $\calP^*$ is cut-free, we are done. Otherwise,
$M(\calP)$ contains a $\forall$-expansion. Let $n\in M(\calP)$
be a $\forall$-expansion such that $\deg(n)$ is maximal in $M(\calP)$. 
Since $\deg(n)$ is maximal and $\calP^*$ is $\lor\land$-normal, no node dominates $n$. 
Hence we may apply the quantifier-reduction rule to $n$, which decreases $o(\calP^*,r)$
by Lemma~\ref{lem:reduction:new}.
At some point, $o(\calP^*,r)=0$, and the next cut-reduction will be applied to a $\forall$-expansion
of rank $<r$. Since, by Lemmas~\ref{lem:lorland_normal}~and~\ref{lem:reduction:new}, 
$\rk(\calP^*)$ never increases, we conclude termination of the strategy by double induction.
Finally, $\Shallow(\calP)=\Shallow(\calP^*)$ by Lemma~\ref{lem:red_pres_proof}.
\end{proof}
\subsection{Strong Normalization}\label{sec:strong_normalization}
%
Having shown weak normalization of the cut-reduction rules in the previous
section, it is important to turn to the question of strong normalization,
i.e.~whether {\em all} reduction sequences are of finite length. We conjecture
that our cut-reduction rules are indeed strongly normalizing, and present some
evidence for this claim by discussing how our reduction rules behave on
a translation of the example~\cite[Figure 14]{Heijltjes10Classical}, which causes
a failure of strong normalization in the setting of proof forests.

This example can be translated as an expansion proof
of the form $\calP=(C_1^+,C_1^-),(C_2^+,C_2^-),\calP'$ (where $\calP'$ is cut-free) with
\[
\begin{array}{lll}
C_1^+=&\exists x\,\forall y\, \dual{P(x,y)} &+^c \forall y\, \dual{P(c,y)} +^\gamma \dual{P(c,\gamma)}\\
      &                                     &+^\gamma \forall y\, \dual{P(\gamma,y)} +^\delta \dual{P(\gamma,\delta)}\\
C_1^-=&\forall x\,\exists y\, P(x,y) &+^\alpha\exists y\,P(\alpha,y) +^\beta P(\alpha,\beta)\\
  \\
C_2^+=&\exists x\,\forall y\, \dual{Q(x,y)} &+^c \forall y\, \dual{Q(c,y)} +^\epsilon \dual{Q(c,\epsilon)}\\
      &                                     &+^\epsilon \forall y\, \dual{Q(\epsilon,y)} +^\iota \dual{Q(\epsilon,\iota)}\\
C_2^-=&\forall x\,\exists y\, Q(x,y) &+^\beta\exists y\,Q(\beta,y) +^\alpha Q(\beta,\alpha).\\
\end{array}
\]
It can be checked that any application of our cut-reduction rules to such a proof
terminates. This is essentially due to the different treatment of bridges (i.e.~dependencies
between different sides of a cut, see Section~\ref{sec:cut_reduction_steps}) in our
formalism: at the core of the non-termination of~\cite[Figure 14]{Heijltjes10Classical}
lies a single bridge~\cite[Figure 16]{Heijltjes10Classical} which induces a cycle. In our
setting, if $\calP$ is an expansion proof containing a single cut, and $\calP\mapsto\calP'$
via a quantifier reduction rule, then $\calP'$ still contains only a single cut.

Indeed, a reduction sequence similar to the non-terminating one described
in~\cite[Figure 17]{Heijltjes10Classical} exists, and it ends in such an expansion proof 
containing only a single cut which is, also in our setting, a bridge. The cut
reduces then to a single propositional cut, the elimination of which is easily 
seen to be strongly normalizing.

In the setting of proof forests, the non-termination due to bridges is handled
by adding a pruning reduction.
One explanation for the fact that in our setting,
we are able to get by without such a reduction, is the use of the merge in the
definition of the cut-reduction rules. The merge has the advantage that it
is very natural, it is an extension of the merge for cut-free expansion 
proofs from~\cite{Miller87Compact}, and it is useful also in applications
not related to cut-elimination, as in the proof of Theorem~\ref{thm:completeness}.
\subsection{Confluence}
It is well-known that cut-elimination (and similar procedures) in classical logic are typically
not confluent, see e.g.~\cite{Urban00Classical,Ratiu12Exploring,Baaz05Experiments}
for case studies and~\cite{Baaz11Nonconfluence,Hetzl12Computational} for asymptotic results.
Neither the proof forests of~\cite{Heijltjes10Classical} nor the Herbrand nets 
of~\cite{McKinley13Proof} have a confluent reduction. The situation is analogous
in our formalism: the reduction is not confluent. In fact, one can use the same example
to demonstrate this; let
\begin{align*}
\calP = & \{ \exists x\, A +^s A\unsubst{x}{s} +^t A\unsubst{x}{t}, \forall x\, \dual{A} +^\alpha \dual{A}\unsubst{x}{\alpha} \},\\
& \{ \exists x\, B +^\alpha B\unsubst{x}{\alpha} , \forall x\, \dual{B} +^\beta \dual{B}\unsubst{x}{\beta} \},\\
& \exists x \exists y\, C +^\alpha ( \exists y\, C\unsubst{x}{\alpha} +^\beta C\sop\sel{x}{\alpha},\sel{y}{\beta}\scl ).
\end{align*}
which is the translation of~\cite[Figure 12]{Heijltjes10Classical} into an expansion
proof with cut. Then it can be verified by a quick calculation that the choice
of reducing either the cut on $A$ or that on $B$ first determines which of
two normal forms is obtained.

However cut-elimination in classical logic can be shown confluent on the level
of the (cut-free) expansion tree on a certain class of proofs~\cite{Hetzl12Herbrand}.
For future work we hope to use such techniques for describing a confluent reduction
in expansion proofs whose normal form is unique and most general in the sense
that it contains all other normal forms as sub-expansions.

\section{Conclusion}

In this paper we have presented expansion proofs with cut
for full first-order logic including non-prenex formulas. Our definitions
extend the existing notion of cut-free expansion proofs in a natural
way. We have given a cut-elimination procedure and proved weak normalization;
strong normalization remains an open problem. Our proof of weak normalization
is inspired by the $\varepsilon$-calculus which allowed to cover also the non-prenex
case without technical difficulties. The complex object-level syntax
of the $\varepsilon$-calculus is avoided in our work by taking care of the
mutual dependencies of variables by the merge operation of expansion trees.

It should be noted that the $\varepsilon$-calculus is, in a sense, more general than expansion-proofs
since there are formulas in the $\varepsilon$-language which do not arise
by translation from usual formulas. But of course, our objective is not to create
a general formalism, but rather to find a good model of cut-elimination for 
the classical first-order sequent calculus! For this purpose, we believe 
expansion proofs with cut are very promising, as they are compact, focus on
the first-order level of proofs, and admit natural cut-reduction rules
which are weakly normalizing --- and perhaps even strongly normalizing.

{\bf Acknowledgements.}
The authors would like to thank D.~Miller and K.~Chaudhuri, M.~Baaz, W.~Heijltjes
and R.~McKinley for many helpful discussions about expansion trees, the
$\varepsilon$-calculus, proof forests and Herbrand nets respectively.

\bibliographystyle{plain}
\bibliography{references}

\vfill
\pagebreak
\section{Appendix}

In this appendix we describe the technical details that have been omitted from
the main paper. Note that the numbering of the results is non-monotonic: For those
results that are stated in the main text, we have retained the numbers, while we
introduce new ones for intermediate results presented exclusively in the appendix.
We also include some examples here that did not fit into the main paper.

\subsection{Basic Operations on Expansion Proofs}

\begin{example}\label{ex:subst}
Consider the following expansion tree $E$ and $E\sigma$ with $\sigma=[\alpha\backslash c]$.
\begin{align*}
E & = \exists x\forall y\, R(x,y) +^c ( \forall y\, R(c,y) +^\beta R(c,\beta)) +^\alpha ( \forall y\, R(\alpha,y) +^\gamma R(\alpha,\gamma))\\
E\sigma & = \exists x\forall y\, R(x,y) +^c ((\forall y\, R(c,y) +^\beta R(c,\beta)) \omerge (\forall y\, R(c,y) +^\gamma R(c,\gamma)))
\end{align*}
Note that the only expansion $w$ in $E$ that has no $w'$ in $E\sigma$ such that $\pred_s(w')=w$ is
the $+^\alpha$ expansion --- it is replaced by a $\omerge$-node.
\end{example}
By induction on the definition, it is easy to see that if $\sigma$ is a renaming
and $\calP$ an expansion tree (without merge), then $\calP\sigma$ is an expansion
tree without merge. 
An important property is that substitution commutes with $\Deep(\cdot)$
and $\Shallow(\cdot)$.
\begin{lemma}\label{lem:subst_deep}
Let $\calP$ be an expansion pre-proof with merges and $\sigma$ a substitution. Then 
\begin{itemize}
\item $\Deep(\calP\sigma)$ is logically equivalent to $\Deep(\calP)\sigma$, and
\item $\Shallow(\calP\sigma) = \Shallow(\calP)\sigma$.
\end{itemize}
\end{lemma}
\begin{proof}
By induction on the structure of $\calP$.
\end{proof}
%
%
\thmrecount{11}{
\begin{lemma}
Let $\calP=\calP'[E]$ be an expansion proof with merges and $\sigma$
a substitution admissible for $\calP$. Then $\calQ=\calP'[E\sigma]$ 
is an expansion proof with merges, and $\Shallow(\calP)=\Shallow(\calQ)$.
\end{lemma}}
\begin{proof}
Lemma~\ref{lem:subst_deep} implies
the latter claim since $\Shallow(\calP)$ does not contain free variables. The same Lemma also
implies that $\Deep(\calQ)$ is a tautology
since propositional tautology-hood is preserved under substitution.
Regularity is preserved since
every subtree of $E$ is ``copied'' exactly once to create $E\sigma$.
To show that $<_{\calQ}$ is acyclic, we show that $w <_{\calQ} w$
implies $\pred_s(w) <_{\calP} \pred_s(w)$.
The only non-trivial case is that there exist an $\exists$-expansion $v$
such that $\pred_s(v)$ has an expansion term containing $\alpha$, and
a $\forall$-expansion $u$ with eigenvariable
$\beta$, such that $w>_{\calQ}v>_{\calQ}u>_{\calQ}w$,
and $\beta\in\Var(\alpha\sigma)$.
But this implies $\pred_s(u)>_{\calP} \pred_s(w)>_{\calP}\pred_s(v)=q(\beta)$,
which contradicts admissibility of $\sigma$.
\end{proof}

\begin{example}\label{ex:mergered}
Continuing Example~\ref{ex:subst}, we have $E\sigma\stackrel{\omerge}{\mapsto}E'\stackrel{\omerge}{\mapsto}E''$
where
\begin{align*}
E' & = \exists x\forall y\, R(x,y)   +^c (\forall y\, R(c,y) +^\beta R(c,\beta) \omerge R(c,\beta))\\
E'' & = \exists x\forall y\, R(x,y)  +^c  (\forall y\, R(c,y) +^\beta R(c,\beta))
\end{align*}
The only expansion $w$ in $E'$ with a non-trivial set $\pred^0_\omerge(w)$ is the 
$+^\beta$ expansion: it has $\pred^0_\omerge(w)=\{+^\beta,+^\gamma\}$. Similarly,
the only $w$ in $E''$ with non-trivial $\pred^0_\omerge(w)$ is $R(c,\beta)$: here,
$\pred^0_\omerge(w)$ consists of the two occurrences of $R(c,\beta)$ in $E'$.
\end{example}
%
\begin{lemma}\label{lem:mergered_props}
The relation $\mergered$ is confluent and strongly
normalizing. Its normal forms have no merge nodes.
\end{lemma}
\begin{proof}
Local confluence follows immediately from the absence of critical pairs.

Let $m$ be a merge node in an expansion pre-proof with merges $\calP$, then the
weight of this node $\weight(m)$ is the number of nodes below it in $\calP$.
Let $m_1,\ldots,m_l$ be the merge nodes in an expansion pre-proof with merges $\calP$,
then the weight of $\calP$ is defined as $\weight(\calP) = \sum_{i=1}^{l} \weight(m_i)$.
The application of eigenvariable renamings and merges decreases the lexicographic
ordering $\langle |\EV(\calP)|, \weight(\calP) \rangle$, hence $\mergered$
is strongly normalizing and we can conclude confluence.

That its normal forms have no more merge nodes is immediate.
\end{proof}

\begin{example}\label{ex:normal_form}
Continuing Example~\ref{ex:mergered}, we have $E''=\normf{E\sigma}$.
There is exactly one expansion $w$ in $E''$ with a non-trivial set $\pred_\omerge(w)$ (containing expansions from $E\sigma$):
namely, $\pred_\omerge(+^\beta)=\{+^\beta,+^\gamma\}$.
\end{example}
An important property is that $\mergered$ preserves the
proof properties of $\calP$.
\begin{lemma}\label{lem:merge_deep}
  If $\calP \stackrel{\omerge}{\mapsto} \calP'$ then
\begin{itemize} 
\item if $\Deep(\calP)$
  is valid then $\Deep(\calP')$ is valid, and
\item if $<_\calP$ is acyclic then $<_{\calP'}$ is acyclic, and
\item if $\calP$ is regular than so is $\calP'$
\end{itemize}
\end{lemma}
\begin{proof}
By inspection of the definition, we show that there exists
a variable renaming $\sigma$ such that $\Deep(\calP)\sigma\impl
\Deep(\calP')$. The variable renaming is used
in case 3. Note that logical equivalence is not preserved due to
case 2. Acyclicity is shown by verifying that for all expansions $v,w$ 
from $\calP'$, if $w>^0_{\calP'} v$ then there exist $w'\in\pred^0_\omerge(w)$
such that for all $v'\in\pred^0_\omerge(v)$ we have $w'>^0_{\calP}v'$. This allows
to translate a cycle from $>_{\calP'}$ to $>_\calP$.
Finally, regularity of
$\calP'$ follows since no new $\forall$-expansions are introduced.
\end{proof}
%
\thmrecount{15}{
\begin{lemma}
  If $\calP_1\omerge\calP_2$ is an expansion proof with merge such that $\Shallow(\calP_1)=\Shallow(\calP_2)$,
  then $\calP_1\cup \calP_2$ is an
  expansion proof and $\Shallow(\calP_1\cup \calP_2)=\Shallow(\calP_1)=\Shallow(\calP_2)$.
\end{lemma}}
\begin{proof}
Since $\Shallow(\calP_1)=\Shallow(\calP_2)$, all non-cut expansion-trees are merged and
we have $\Shallow(\calP_1\cup \calP_2)=\Shallow(\calP_1)=\Shallow(\calP_2)$ by definition.
The proof-properties of $\calP_1\omerge\calP_2$ are carried over to $\calP_1\cup\calP_2$
by Lemma~\ref{lem:merge_deep}~and~Lemma~\ref{lem:subst_proof}: in case 3 of Definition~\ref{def:omerge_reduction},
the eigenvariable renaming is admissible since $q(\alpha_1)$ and $q(\alpha_2)$ are dominated by the same
expansions, and are contained in the same cut (if any).
Finally, $\calP_1\cup\calP_2$ does not contain merge nodes
by Lemma~\ref{lem:mergered_props}.
\end{proof}

\subsection{Cut-Elimination}

Towards verifying that $\mapsto$ is really a binary relation on
expansion proofs, as claimed, we have to prove a technical result on the
behavior of the merge w.r.t.~cut-reduction, namely that those $\forall$-expansions
that do not depend upon the reduced $\forall$-expansion are merged.

\begin{lemma}\label{lem:merge_top_evs}
Let $\calP\mapsto\calP'$ by the quantifier-reduction rule.
We write $\calP'=\normf{\calP''}$ where $\calP''$ is the
expansion tree with merge constructed by the reduction rule.
Let $w$ be the $\forall$-expansion indicated by the rule, and
let $v$ be a $\forall$-expansion in $\calP$ with $w\not<v$. Then there exists
a $\forall$-expansion $u$ in $\calP'$ s.t.~$v'\in \pred_\omerge(u)$ for all copies $v'$ of $v$ in $\calP''$.
\end{lemma}
\begin{proof}
By induction on the merge-reduction sequence. The assumption $w\not<v$ ensures
that in case 4 of Definition~\ref{def:omerge_reduction}, the subtrees containing
copies of $u$ will be merged since the $\exists$-expansions dominating them belong
to the $+^{r_i}$-part.
\end{proof}
%
%
%
\thmrecount{22}{
\begin{lemma}
If $\calP_1\mapsto\calP_2$ and $\calP_1$ is an expansion proof, then $\calP_2$ is an expansion
proof. Furthermore, $\Shallow(\calP_1)=\Shallow(\calP_2)$.
\end{lemma}}
\begin{proof}
We only give the proof for the quantifier cut-reduction step; the proof for the other reduction steps
is analogous and simpler.
Let $\sigma_i=\unsubst{\alpha}{t_i}$ and
assume 
\[
  \begin{array}{rl}
     \calP_1 =&  \{ \exists x\, A +^{t_1} E_1 \cdots +^{t_n} E_n , \forall x\, \bar{A} +^\alpha E \}, \calP\\
\mapsto&
\calP \cup \{ E_1\lor\cdots\lor E_n, E\eta_1\sigma_1\land\cdots\land E\eta_n\sigma_n \} \cup \bigcup_{i=1}^n \calP\eta_i\sigma_i=\calP_2,
\end{array}
\]
where $\eta_i$ are renamings establishing regularity. First, note that the $\eta_i$ 
are trivially admissible for $\calP_1$ since only new variables are introduced.
Next, we show that the $\sigma_i$ are admissible for
$\calP \cup \{ E_1\lor\cdots\lor E_n, E\eta_1\land\cdots\land E\eta_n \} \cup \bigcup_{j=1}^n \calP\eta_j$.
Hence assume $\beta\in\Var(t_i)$ and that there exists an $\exists$-expansion $w$ in $\calP\eta_j$
or $E\eta_j$ with expansion term $t$ such that $\alpha\in\Var(t)$ and that $w<q(\beta)$. This is
only possible if there is a $\forall$-expansion $v$ with eigenvariable $\gamma$
such that $v<q(\beta)$ and $w$ dominates $v$. But since $\gamma$ is a fresh variable introduced
by $\eta_j$, this
implies that $v$ dominates $q(\beta)$, hence $\beta$ is also a fresh variable introduced by
$\eta_j$, which contradicts $\beta\in\Var(t_i)$.

Towards showing that $\calP_2$ is an expansion proof, we have to make the definition of expansion proof
with merge slightly more liberal: we allow cuts $\{C^+,C^-\}$ such that
$\Shallow(C^+)=\dual{\Shallow(C^-)}\eta$, where $\eta$ is a renaming (the usual
definition requires $\eta$ to be the identity renaming). 
The results from Section~\ref{subsec:merge},
hold as well for this definition. $\eta$ will be chosen such that after merge-normalization,
all cuts will be syntactically correct.

Now we show that $\calP_2$ is an expansion proof:
writing $\calP_2=\normf{\calP_2'}$,
by Lemma~\ref{lem:subst_proof} and Lemma~\ref{lem:merge_proof} it suffices to show that 
$\calP_2'$
is regular (which it is by construction), that its dependency relation is acyclic,
that $\Deep(\calP_2')$ is a tautology, and that $\Shallow(\calP_2')=\Shallow(\calP_1)$ (which holds
by construction as well).

To show that $\Deep(\calP_2')$ is valid, we start by reducing the problem:
It can be checked (using the propositional tautology $(A \land A')\lor (B \land B') \impl (A\lor B)\land (A'\lor B')$)
that $\Deep(\calP_2')$ is implied by
\[
  F=\Deep(\calP)\lor(\bigvee_{i=1}^n\Deep(E_i))\land(\bigwedge_{i=1}^n\Deep(E\eta_i\sigma_i))\lor\bigvee_{i=1}^n\Deep(\calP\eta_i\sigma_i).
\]
By Lemma~\ref{lem:subst_deep}, $F$ is logically equivalent
to 
\[
  F'=\Deep(\calP)\lor(\bigvee_{i=1}^n\Deep(E_i))\land(\bigwedge_{i=1}^n\Deep(E)\eta_i\sigma_i)\lor\bigvee_{i=1}^n\Deep(\calP)\eta_i\sigma_i.
\]
Hence it suffices to show that $F'$ is valid.

Note that 
$\Deep(\calP_1)=\Deep(\calP)\lor(\bigvee_{i=1}^n\Deep(E_i)\land \Deep(E))$.
Since $\Deep(\calP_1)$ is valid, the formulas
$\Deep(\calP)\eta_i\sigma_i\lor\Deep(E)\eta_i\sigma_i$ and
$\Deep(\calP)\lor\bigvee_{i=1}^n\Deep(E_i)$
are valid. Using propositional reasoning, in particular validity of
$(A\lor B)\land (C\lor D) \impl A\lor (B \land C)\lor D$, we obtain
validity of $F'$.

Next, we show that acyclicity of $<_{\calP_1}$ implies
acyclicity of $<_{\calP_2}$. This follows from the fact
that if $x<_{\calP_2}y$ implies that there exist $x'\in\pred_c(x),
y'\in\pred_c(y)$ such that $x' <_{\calP_1} y'$. Hence a cycle
in $<_{\calP_2}$ gives rise to a cycle in $<_{\calP_1}$.

Finally, we have to show that all cuts in $\calP_2$ are syntactically correct.
By construction, the only ``incorrect'' cut in $\calP_2'$ is the indicated one.
We have $\Shallow(E\sigma_i)=\Shallow(E\eta_i\eta_i^{-1}\sigma_i)=\Shallow(E\eta_i\sigma_i)\eta_i^{-1}$
since $\eta_i$ is a renaming to fresh variables.
Since $\Shallow(E)\sigma_i=\dual{\Shallow(E_i)}$, this yields $\dual{\Shallow(E_i)}\eta_i=\Shallow(E\eta_i\sigma_i)$, hence
the ``incorrect'' cut fulfills the liberalized definition. 
We even have $\dual{\Shallow(E_i)}\eta_i=\dual{\Shallow(E_i)}\eta_i'$
for a variable renaming $\eta_i'$ such that if $\beta\in\dom(\eta_i')$
then $q(\alpha)\not< q(\beta)$. For if $\beta\in\Var(\Shallow(E_i))$
and $\beta\notin\Var(\Shallow(E\eta_i\sigma_i))$
then $\beta\in\Var(\Shallow(c))$ where 
$c$ is the indicated cut in $\calP_1$, and therefore
$q(\beta)<_{\calP_1}q(\alpha)$. Since $<_{\calP_1}$ is acyclic, we 
have $q(\alpha)\not<q(\beta)$. Hence we can take for $\eta_i'$
just $\eta_i$ where these $\beta$ are not renamed.
Finally, Lemma~\ref{lem:merge_top_evs} implies that the copies
of the variables in $\dom(\eta_i')$ are identified by the merge,
which yields correctness of the cuts in $\calP_2$.
\end{proof}
\begin{example}\label{ex:rk_deg}
  Consider an expansion proof with three cuts $\calP=C_1,C_2,C_3,\calE$ where $C_i=\{C_i^+,C_i^-\}$ for
  $1\leq i \leq 3$ where
\begin{align*}
C_1^+ & = \exists x\forall y\, P(x,y)
    +^c ( \forall y\, P(c,y)
       +^\gamma P(c,\gamma) ) \\
C_1^- & = \forall x\exists y\, \dual{P}(x,y)
    +^\alpha ( \exists y\, \dual{P}(\alpha,y)
       +^c \dual{P}(\alpha,c) ) \\
C_2^+ & = \exists x Q(\alpha,x)
    +^\alpha Q(\alpha,\alpha)
       +^c Q(\alpha,c)\\
C_2^- & = \forall x Q(\alpha,x)
       +^\beta \dual{Q(\alpha,\beta)}\\
C_3^+ & = \exists x\forall y\, R(x,y)
    +^\beta ( \forall y\, P(\beta,y)
       +^\lambda P(\beta,\lambda) ) \\
C_3^- & = \forall x\exists y\, \dual{R}(x,y)
    +^\delta ( \exists y\, \dual{R}(\delta,y)
    +^c \dual{R}(\delta,c) +^\alpha \dual{R}(\delta,\alpha))
\end{align*}
Assuming that $\calE$ is cut-free and contains no $\forall$-nodes,
it is of no importance in this context, and so we do not
give its definition.
Denote the expansions in these trees from left to right,
top to bottom, by $w_1,\ldots,w_{12}$
(i.e.~$w_1$ is the $+^c$ expansion in $C_1^+$, $w_{12}$ is 
the $+^\alpha$ expansion in $C_3^-$, etc). 
Then the maximal $>$-chain descending from $w_9$ is
$w_9>w_8>w_7>w_3$, yielding $\deg(w_9)=3,\deg(w_8)=2,\deg(w_7)=1,\deg(w_3)=0$.
In fact, $w_9$ is the node of maximal degree in $\calP$.
Furthermore,
$\rk(w_i)=1$ and $\rk(w_j)=2$ for $i\in\{2,4,5,6,7,9,11,12\}$ and $j\in\{1,3,8,10\}$. 
\end{example}
%
\thmrecount{27}{
\begin{lemma}
  For every expansion proof $\calP$ there is a $\lor\land$-normal expansion proof $\calP^*$
  such that $\calP \rightarrow \calP^*$, $\Shallow(\calP^*) = \Shallow(\calP)$,
  $\rk(\calP^*)=\rk(\calP)$
  and 
  $o(\calP^*,r)=o(\calP,r)$ for all $r$.
\end{lemma}}
\begin{proof}
  If $\calP$ is not $\lor\land$-normal, then $\calP=(E_1 \lor E_2, E'_1\land E'_2), \calP'$
  and hence 
  \[
    \calP\mapsto (E_1, E'_1), (E_2, E'_2), \calP'=\calP^*.
  \]
  Since the number
  of $\lor\land$-cuts in $\calP^*$ is strictly smaller than the number of $\lor\land$-cuts in $\calP$,
  and $\rk(w)=\rk(\pred_c(w))$ for all expansions $w$ in $\calP^*$,
  we conclude by induction.
\end{proof}
Let $(E_1, E_2)$ be a cut in an expansion proof. Since $\Shallow(E_1)=\dual{\Shallow(E_2)}$
we can associate in a natural way to every expansion $w$
in $E_1$ a non-empty set of dual expansions in $E_2$, the set of 
{\em dual expansions} $\cl(w)$. This association is symmetric, i.e.~$v\in\cl(w)$ exactly if $w\in\cl(v)$.
\begin{example}
  We continue Example~\ref{ex:rk_deg}, giving the sets
  of dual expansions for $C_3$: $\cl(w_8)=\{w_{10}\}$,
  $\cl(w_9)=\{w_{11},w_{12}\}$, $\cl(w_{10})=\{w_8\}$, $\cl(w_{11})=\{w_9\}$,
  $\cl(w_{12})=\{w_9\}$.
\end{example}
\begin{lemma}\label{lem:equal_rank}
  Let $w$ be a critical expansion. Then
  $\rk(w)=\rk(w')$ for all $w'\in\cl(w)$.
\end{lemma}
\begin{proof}
  By structural induction on the cut-formula, noting that since $w,w'$ occur in
  the same cut, they have the same cut-formula.
\end{proof}
\begin{lemma}\label{lem:have_max_forall}
  If $\calP$ contains a critical expansion, then $M(\calP)$ contains a $\forall$-expansion.
\end{lemma}
\begin{proof}
  Since $M(\calP)$ is non-empty, the result follows
  from Lemma~\ref{lem:equal_rank}.
\end{proof}
%
%
The following result is a trivial consequence of the definition,
and will ensure that a cut-reduction rule is applicable to expansions
in $M(\calP)$.
\begin{lemma}\label{lem:uppermost}
  Let $\calP$ be an expansion proof and $w\in M(\calP)$.
  Then no expansion dominates $w$.
\end{lemma}
%
\thmrecount{28}{
\begin{lemma}
  Let $w,v$ be quantifier nodes in $\normf{\calP[\alpha\backslash t]}$ such that
  $\alpha$ does not occur in any cut-formula in $\calP$, and let
  $w'\in\pred_\omerge(w)$ and $v'\in\pred_\omerge(v)$. Then $w$ dominates $v$ if and only if 
  $w'$ dominates $v'$. Furthermore, $\rk(w')=\rk(w)$.
\end{lemma}}
\begin{proof}
  By induction on the definition of $\calP\sigma$,
  it is easy to show that $w$ dominates $v$ in $\calP\sigma$ iff 
  $\pred_s(w)$ dominates $\pred_s(v)$.
  Next, we show that if $\calP_1 \stackrel{\omerge}{\mapsto} \calP_2$
  then for $w,v$ nodes in $\calP_2$ and $w'\in\pred^0_\omerge(w),v'\in\pred^0_\omerge(v)$, $w$ dominates $v$
  iff $w'$ dominates $v$.
  This is obvious in case 2 of the definition. In case 1, we
  have $E[L \omerge L] \stackrel{\omerge}{\mapsto} E[L]$, and it suffices
  to observe that a node $w$ dominates $L$ in $E[L]$ iff all $w'\in\pred^0_\omerge(w)$
  dominate both occurrences of $L$ in $E[L \omerge L]$.
  In cases 3 and 4, we reason analogously, using the result for $\pred_s$ we just proved
  for case 3.
  Finally, we extend the result to $\pred_\omerge$ by induction on its definition.
  $\rk(w')=\rk(w)$ follows immediately from the first statement.
\end{proof}
%
\thmrecount{30}{
\begin{lemma}
Let $r=\rk(\calP)$ and $\sigma$ a substitution. Then $o(\calP\sigma,r)=o(\calP,r)$. 
Furthermore, let $E_1,E_2$ be expansion trees and
$E=\normf{E_1\omerge E_2}$,
then $o(E,r)= o(E_1,r)=o(E_2,r)$.
\end{lemma}}
\begin{proof}
$o(\calP\sigma,r)=o(\calP,r)$ holds for all $r$ by Lemma~\ref{lem:subst_rank_preserve}.
Let $v$ be a $\forall$-node of rank $r$ in $E$. Then there is a $w\in\pred_\omerge(v)$
such that $\rk(w)=r$, and since $\Shallow(E_1)=\Shallow(E_2)$
and by Lemma~\ref{lem:uppermost}, there is a unique
$w'$ corresponding to $w$ in $E_2$.
It is then easy to see by induction on
an appropriate $\mergered$-sequence that $\pred_\omerge(v)=\{w,w'\}$. From this, the
claim follows.
%
%
\end{proof}
We conclude by giving a more detailed proof of the
weak-normalization result.
%
%
\thmrecount{31}{
\begin{lemma}
Let $\calP_1\mapsto\calP_2$ by the quantifier reduction-rule, 
let $r=\rk(\calP_1)$ and denote
the reduced $\forall$-expansion by $w$. 
If $w\in M(\calP)$ and $\deg(w)$ is maximal in $M(\calP_1)$, then
$\rk(\calP_2)\leq r$
and $o(\calP_2,r)= o(\calP_1,r)-1$.
\end{lemma}}
\begin{proof}
We have $\rk(\calP_2)\leq r$ since the rank changes for no expansions by
Lemmas~\ref{lem:subst_rank_preserve}~and~\ref{lem:merge_deg}.
To show that $o(\calP_2,r)= o(\calP_1,r)-1$, it suffices to show that
for all non-reduced cuts $G\in\calP_1$ containing a $\forall$-expansion of rank $r$,
$\alpha\notin\Var(\Shallow(G))$ and if $\beta\in\Var(\Shallow(G))$ then
$w\not<q(\beta)$: If this is so, then $\Shallow(G)\unsubst{\alpha}{t_i}=\Shallow(G)$ 
and 
by Lemma~\ref{lem:merge_top_evs}, 
the regularization $\eta_i$ is reversed w.r.t.~$\Shallow(G)$
by the merge, and hence the cuts are merged. Therefore, by Lemma~\ref{lem:order_subst_merge},
their order stays the same. Furthermore, $E_i,E$ do not contain expansions of maximal
rank (since $w$ has maximal rank), and $w$ does not have a successor in $\calP_2$,
hence $o(\calP_2,r)= o(\calP_1,r)-1$.

To show the claim, consider a $\forall$-expansion $v$ of rank $r$ in $G$. 
If 
$\alpha\in\Var(\Shallow(G))$, then $v>q(\alpha)=w$ and
therefore $\deg(v)>\deg(w)$, which contradicts maximality of $\deg(w)$. 
Similarly, $\beta\in\Var(\Shallow(G))$ implies
$q(\beta)<w$. Assuming $w<q(\beta)$ yields $w<v$ and again the contradictory
$\deg(w)<\deg(v)$.
\end{proof}
%
%
\thmrecount{32}{
\begin{theorem}[Weak Normalization]
For every expansion proof $\calP$ there is a cut-free expansion proof $\calP^*$
with $\Shallow(\calP) = \Shallow(\calP^*)$ and $\calP \rightarrow \calP^*$.
\end{theorem}}
\begin{proof}
First, we apply the propositional cut-reduction
rules exhaustively to $\calP$ to obtain an $\lor\land$-normal expansion proof $\calP^*$
(Lemma~\ref{lem:lorland_normal}). If $\calP^*$ is cut-free, we are done. Otherwise,
$M(\calP)$ contains a $\forall$-expansion by Lemma~\ref{lem:have_max_forall}. Let $n\in M(\calP)$
be a $\forall$-expansion such that $\deg(n)$ is maximal in $M(\calP)$. 
Since $\deg(n)$ is maximal and $\calP^*$ is $\lor\land$-normal, no node dominates $n$
by Lemma~\ref{lem:uppermost}. 
Hence we may apply the quantifier-reduction rule to $n$, which decreases $o(\calP^*,r)$
by Lemma~\ref{lem:reduction:new}.
At some point, $o(\calP^*,r)=0$, and the next cut-reduction will be applied to a $\forall$-expansion
of rank $<r$. Since, by Lemmas~\ref{lem:lorland_normal}~and~\ref{lem:reduction:new}, 
$\rk(\calP^*)$ never increases, we conclude termination of the strategy by double induction.
Finally, $\Shallow(\calP)=\Shallow(\calP^*)$ by Lemma~\ref{lem:red_pres_proof}.
\end{proof}
\end{document}